\documentclass[twocolumn]{article}
\usepackage{tikz}
\usepackage{blkarray}
\usetikzlibrary{calc}
\usepackage[bookmarksopen=true]{hyperref}

\input{nqubit.sty}

\title{Exact synthesis of multiqubit Clifford+$T$ circuits}

\author{\begin{tabular}{c}
    Brett Giles \\[.5ex]
    \normalsize Department of Computer Science \\
    \normalsize University of Calgary
  \end{tabular}\begin{tabular}{c}
    Peter Selinger \\[.5ex]
    \normalsize Department of Mathematics and Statistics \\
    \normalsize Dalhousie University
  \end{tabular}
}

\date{\begin{minipage}{0.9\textwidth}\normalsize
  We prove that a unitary matrix has an exact representation over the
  Clifford+$T$ gate set with local ancillas if and only if its entries
  are in the ring $\Z[\frac1{\sqrt{2}},i]$. Moreover, we show that one
  ancilla always suffices. These facts were conjectured by
  Kliuchnikov, Maslov, and Mosca. We obtain an algorithm for
  synthesizing a exact Clifford+$T$ circuit from any such $n$-qubit
  operator. We also characterize the Clifford+$T$ operators that can
  be represented without ancillas.
\end{minipage}}

\begin{document}
\maketitle

\section{Introduction}

An important problem in quantum information theory is the
decomposition of arbitrary unitary operators into gates from some
fixed universal set {\cite{Nielsen-Chuang}}. Depending on the operator
to be decomposed, this may either be done exactly or to within some
given accuracy $\epsilon$; the former problem is known as {\em exact
  synthesis} and the latter as {\em approximate synthesis}
{\cite{Kliuchnikov-et-al}}. 

Here, we focus on the problem of exact synthesis for $n$-qubit
operators, using the Clifford+$T$ universal gate set. Recall that the
Clifford group on $n$ qubits is generated by the Hadamard gate $H$,
the phase gate $S$, the controlled-not gate, and the scalar
$\omega=e^{i\pi/4}$ (one may allow arbitrary unit scalars, but it is
not convenient for our purposes to do so). It is well-known that one
obtains a universal gate set by adding the non-Clifford operator $T$
{\cite{Nielsen-Chuang}}.
\begin{equation}\label{eqn-generators}
  \begin{array}{c}
  \displaystyle
  \omega = e^{i\pi/4},\sep
  H = \frac{1}{\sqrt{2}}\zmatrix{cc}{1&1\\1&-1},\sep 
  S = \zmatrix{cc}{1&0\\0&i},\\\\[-1.5ex]
  \displaystyle
  {\it CNOT} = \zmatrix{cccc}{1&0&0&0\\0&1&0&0\\0&0&0&1\\0&0&1&0},\sep
  T = \zmatrix{cc}{1&0\\0&e^{i\pi/4}}.
\end{array}
\end{equation}
In addition to the Clifford+$T$ group on $n$ qubits, as defined above,
we also consider the slightly larger group of Clifford+$T$ operators
``with ancillas''. We say that an $n$-qubit operator $U$ is a
Clifford+$T$ operator {\em with ancillas} if there exists $m\geq 0$
and a Clifford+$T$ operator $U'$ on $n+m$ qubits, such that
$U'(\ket{\phi}\otimes\ket{0})= (U\ket{\phi})\otimes\ket{0}$ for all
$n$-qubit states $\ket{\phi}$.

Kliuchnikov, Maslov, and Mosca {\cite{Kliuchnikov-et-al}} showed that
a single-qubit operator $U$ is in the Clifford+$T$ group if and only
if all of its matrix entries belong to the ring
$\Z[\frac1{\sqrt{2}},i]$. They also showed that the Clifford+$T$
groups ``with ancillas'' and ``without ancillas'' coincide for $n=1$,
but not for $n\geq 2$.  Moreover, Kliuchnikov et al.~conjectured that
for all $n$, an $n$-qubit operator $U$ is in the Clifford+$T$ group
with ancillas if and only if its matrix entries belong to
$\Z[\frac1{\sqrt{2}},i]$. They also conjectured that a single ancilla
qubit is always sufficient in the representation of a Clifford+$T$
operator with ancillas.  The purpose of this paper is to prove these
conjectures. In particular, this yields an algorithm for exact
Clifford+$T$ synthesis of $n$-qubit operators. We also obtain a
characterization of the Clifford+$T$ group on $n$ qubits without
ancillas.

It is important to note that, unlike in the single-qubit case, the
circuit synthesized here are not in any sense canonical, and very far
from optimal. Thus, the question of {\em efficient} synthesis is not
addressed here.

\section{Statement of the main result}

Consider the ring $\Z[\frac1{\sqrt{2}},i]$, consisting of complex
numbers of the form
\[  \frac{1}{2^n} (a + bi + c\sqrt{2} + di\sqrt{2}),
\]
where $n\in\N$ and $a,b,c,d\in\Z$. Our goal is to prove the following
theorem, which was conjectured by Kliuchnikov et
al.~{\cite{Kliuchnikov-et-al}}:
\begin{theorem}\label{thm-main}
  Let $U$ be a unitary $2^n\times 2^n$ matrix. Then the following are
  equivalent:
\begin{enumerate}\alphalabels
\item[(a)] $U$ can be exactly represented by a quantum circuit over
  the Clifford+$T$ gate set, possibly using some finite number of
  ancillas that are initialized and finalized in state $\ket0$.
\item[(b)] The entries of $U$ belong to the ring
  $\Z[\frac1{\sqrt{2}},i]$.
\end{enumerate}
Moreover, in (a), a single ancilla is always sufficient.
\end{theorem}

\section{Some algebra}

We first introduce some notation and terminology, following
{\cite{Kliuchnikov-et-al}} where possible. Recall that $\N$ is the set
of natural numbers including 0, and $\Z$ is the ring of \linebreak 
integers. We write $\Zb=\Z/2\Z$ for the ring of integers modulo 2. Let
$\D$ be the ring of {\em 
  dyadic fractions}, defined as $\D = \Z[\frac12] =
\s{\frac{a}{2^n}\mid a\in\Z, n\in\N}$.

Let $\omega = e^{i\pi/4} = (1+i)/\sqrt{2}$. Note that $\omega$ is an
8th root of unity satisfying $\omega^2=i$ and $\omega^4=-1$. We will
consider three different rings related to $\omega$:

\begin{definition}
  Consider the following rings. Note that the first two are subrings
  of the complex numbers, and the third one is not:
  \begin{itemize}
  \item $\Dw = \s{a\omega^3+b\omega^2+c\omega+d \mid a,b,c,d\in\D}$.
  \item $\Zw = \s{a\omega^3+b\omega^2+c\omega+d \mid a,b,c,d\in\Z}$.
  \item $\Zbw = \s{p\omega^3+q\omega^2+r\omega+s \mid p,q,r,s\in\Zb}$.
  \end{itemize}
  Note that the ring $\Zbw$ only has 16 elements. The laws of addition
  and multiplication are uniquely determined by the ring axioms and
  the property $\omega^4=1\ (\mymod 2)$. We call the elements of
  $\Zbw$ {\em residues} (more precisely, residue classes of $\Zw$
  modulo 2).
\end{definition}

\begin{remark}
  The ring $\Dw$ is the same as the ring $\Z[\frac1{\sqrt{2}},i]$
  mentioned in the statement of Theorem~\ref{thm-main}. However, as
  already pointed out in {\cite{Kliuchnikov-et-al}}, the formulation
  in terms of $\omega$ is far more convenient algebraically.
\end{remark}

\begin{remark}
  The ring $\Zw$ is also called the {\em ring of algebraic integers}
  of $\Dw$. It has an intrinsic definition, i.e., one that is
  independent of the particular presentation of $\Dw$. Namely, a
  complex number is called an {\em algebraic integer} if it is the
  root of some polynomial with integer coefficients and leading
  coefficient 1. It follows that $\omega$, $i$, and $\sqrt{2}$ are
  algebraic integers, whereas, for example, $1/\sqrt{2}$ is not. The
  ring $\Zw$ then consists of precisely those elements of $\Dw$ that
  are algebraic integers.
\end{remark}

\subsection{Conjugate and norm}

\begin{remark}[Complex conjugate and norm]
  Since $\Dw$ and $\Zw$ are subrings of the complex numbers, they
  inherit the usual notion of complex conjugation. We note that
  $\omega\da = -\omega^3$. This yields the following formula:
  \begin{equation}\label{eqn-adjoint}
    (a\omega^3+b\omega^2+c\omega+d)\da = -c\omega^3-b\omega^2-a\omega+d.
  \end{equation}
  Similarly, the sets $\Dw$ and $\Zw$ inherit the usual norm from the complex
  numbers. It is given by the following explicit formula, for
  $t=a\omega^3+b\omega^2+c\omega+d$:
  \begin{equation}\label{eqn-cplx-norm}
    \snorm{t}^2 = t\da t = (a^2+b^2+c^2+d^2) + (cd+bc+ab-da) \sqrt{2}.
  \end{equation}
\end{remark}

\begin{definition}[Weight]
  For $t\in\Dw$ or $t\in\Zw$, the {\em weight} of $t$ is denoted 
  $\sweight{t}$, and is given by:
  \begin{equation}\label{eqn-weight}
    \sweight{t}^2 = a^2+b^2+c^2+d^2.
  \end{equation}
\end{definition}

Note that the square of the norm is valued in $\D[\sqrt{2}]$, whereas
the square of the weight is valued in $\D$.  We also extend the
definition of norm and weight to vectors in the obvious way: For $u =
(u_\jay)_\jay$, we define
\[ \norm{u}^2 = \sum_\jay\norm{u_\jay}^2
\quad\mbox{and}\quad
\weight{u}^2 = \sum_\jay\weight{u_\jay}^2.
\]

\begin{lemma}\label{lem-ring-norm}
  Consider a vector $u\in\Dw^n$. If $\norm{u}^2$ is an integer, then
  $\weight{u}^2=\norm{u}^2$.
\end{lemma}

\begin{proof}
  Any $t\in\D[\sqrt{2}]$ can be uniquely written as $t=a+b\sqrt{2}$,
  where $a,b\in\D$. We can call $a$ the {\em dyadic part} of $t$. Now
  the claim is obvious, because $\weight{u}^2$ is exactly the dyadic
  part of $\norm{u}^2$.
\end{proof}

\subsection{Denominator exponents}

\begin{definition}
  Let $t\in\Dw$. A natural number $k\in\N$ is called a {\em
    denominator exponent} for $t$ if $\sqrt{2}^k t \in \Zw$.  It is
  obvious that such $k$ always exists. The least such $k$ is called
  the {\em least denominator exponent} of $t$.

  More generally, we say that $k$ is a denominator exponent for a
  vector or matrix if it is a denominator exponent for all of its
  entries. The least denominator exponent for a vector or matrix is
  therefore the least $k$ that is a denominator exponent for all of
  its entries.
\end{definition}

\begin{remark}
  Our notion of least denominator exponent is almost the same as the
  ``smallest denominator exponent'' of {\cite{Kliuchnikov-et-al}},
  except that we do not permit $k<0$.
\end{remark}

\subsection{Residues}

\begin{remark}
  The ring $\Zbw$ is not a subring of the complex numbers; rather, it
  is a quotient of the ring $\Zw$. Indeed, consider the {\em parity
    function} $\parity{()}:\Z\to\Zb$, which is the unique ring
  homomorphism. It satisfies $\parity{a}=0$ if $a$ is even and
  $\parity{a}=1$ if $a$ is odd. The parity map induces a surjective
  ring homomorphism $\rho:\Zw\to\Zbw$, defined by
  \[ \rho(a\omega^3+b\omega^2+c\omega+d) =
  \parity{a}\omega^3+\parity{b}\omega^2+\parity{c}\omega+\parity{d}.
  \]
  We call $\rho$ the {\em residue map}, and we call $\rho(t)$ the {\em
    residue} of $t$.
\end{remark}

\begin{convention}
  Since residues will be important for the constructions of this
  paper, we introduce a shortcut notation, writing each residue
  $p\omega^3+q\omega^2+r\omega+s$ as a string of binary digits $pqrs$.
\end{convention}

What makes residues useful for our purposes is that many important
operations on $\Zw$ are well-defined on residues.  Here, we say that
an operation $f:\Zw\to\Zw$ is {\em well-defined on residues} if for
all $t,s$, $\rho(t)=\rho(s)$ implies $\rho(f(t))=\rho(f(s))$.

For example, two operations that are obviously well-defined on
residues are complex conjugation, which takes the form $(pqrs)\da =
rqps$ by (\ref{eqn-adjoint}), and multiplication by $\omega$, which is
just a cyclic shift $\omega(pqrs)=qrsp$. Table~\ref{tab-residue} shows
two other important operations on residues, namely multiplication by
$\sqrt{2}$ and the squared norm.

\begin{table}
  \[ \begin{array}{c|c|c} \rho(t) & \rho(\sqrt{2}\,t) & \rho(t\da t)\\\hline
    0000 & 0000 & 0000 \\
    0001 & 1010 & 0001 \\
    0010 & 0101 & 0001 \\
    0011 & 1111 & 1010 \\

    0100 & 1010 & 0001 \\
    0101 & 0000 & 0000 \\
    0110 & 1111 & 1010 \\
    0111 & 0101 & 0001 \\
  \end{array}\qquad
  \begin{array}{c|c|c} \rho(t) & \rho(\sqrt{2}\,t) & \rho(t\da t)\\\hline
    1000 & 0101 & 0001 \\
    1001 & 1111 & 1010 \\
    1010 & 0000 & 0000 \\
    1011 & 1010 & 0001 \\

    1100 & 1111 & 1010 \\
    1101 & 0101 & 0001 \\
    1110 & 1010 & 0001 \\
    1111 & 0000 & 0000 \\
  \end{array}
  \]
  \caption{Some operations on residues}\label{tab-residue}
\end{table}

\begin{definition}[$k$-Residue]
  Let $t\in\Dw$ and let $k$ be a (not necessarily least) denominator
  exponent for $t$.  The {\em $k$-residue of $t$}, in symbols
  $\rho_k(t)$, is defined to be
  \[ \rho_k(t) = \rho(\sqrt{2}^k t).
  \]
\end{definition}

\begin{definition}[Reducibility]
  We say that a residue $x\in\Zbw$ is {\em reducible} if it is of the
  form $\sqrt{2}\,y$, for some $y\in\Zbw$. Moreover, we say that
  $x\in\Zbw$ is {\em twice reducible} if it is of the form $2y$, for
  some $y\in\Zbw$.
\end{definition}

\begin{lemma}\label{lem-reducible}
  For a residue $x$, the following are equivalent:
  \begin{enumerate}\alphalabels
  \item $x$ is reducible;
  \item $x\in\s{0000,0101,1010,1111}$;
  \item $\sqrt{2}\,x = 0000$;
  \item $x\da x=0000$.
  \end{enumerate}
  Moreover, $x$ is twice reducible iff $x=0000$.
\end{lemma}

\begin{proof}
  By inspection of Table~\ref{tab-residue}.
\end{proof}

\begin{lemma}\label{lem-reducible2}
  Let $t\in\Zw$. Then $t/2\in\Zw$ if and only if $\rho(t)$ is twice
  reducible, and $t/\sqrt{2}\in\Zw$ if and only if $\rho(t)$ is
  reducible.
\end{lemma}

\begin{proof}
  The first claim is trivial, as $\rho(t)=0000$ if and only if all
  components of $t$ are even. For the second claim, the left-to-right
  implication is also trivial: assume $t'=t/\sqrt{2}\in\Zw$. Then
  $\rho(t) = \rho(\sqrt{2}\,t')$, which is reducible by
  definition. Conversely, let $t\in\Zw$ and assume that $\rho(t)$ is
  reducible. Then $\rho(t)\in\s{0000, 0101, 1010, 1111}$, and it
  can be seen from Table~\ref{tab-residue} that
  $\rho(\sqrt{2}\,t)=0000$. Therefore, $\sqrt{2}\,t$ is twice reducible by
  the first claim; hence $t$ is reducible.
\end{proof}

\begin{corollary}\label{cor-reducible3}
  Let $t\in\Dw$ and let $k>0$ be a denominator exponent for $t$. Then
  $k$ is the least denominator exponent for $t$ if and only if
  $\rho_k(t)$ is irreducible.
\end{corollary}

\begin{proof}
  Since $k$ is a denominator exponent for $t$, we have $\sqrt{2}^k t
  \in \Zw$. Moreover, $k$ is least if and only if $\sqrt{2}^{k-1}
  t\not\in\Zw$. By Lemma~\ref{lem-reducible2}, this is the case if and
  only if $\rho(\sqrt{2}^k t)=\rho_k(t)$ is irreducible.
\end{proof}

\begin{definition}
  The notions of residue, $k$-residue, reducibility, and
  twice-reducibility all extend in an obvious componentwise way to
  vectors and matrices. Thus, the residue $\rho(u)$ of a vector or
  matrix $u$ is obtained by taking the residue of each of its entries,
  and similar for $k$-residues. Also, we say that a vector or matrix
  is reducible if each of its entries is reducible, and similarly for
  twice-reducibility.
\end{definition}

\begin{example}\label{exa-k-residue}
  Consider the matrix
  \[ \small U \,{=}\, \frac{\small 1}{\small \sqrt{2}^3}\!\footnotesize\zmatrix{cccc}{
    -\omega^3+\omega -1          
			& \omega^2+\omega+1 
				& \omega^2 
					& -\omega\\
    \omega^2+\omega 
			& -\omega^3+\omega^2 
				& -\omega^2-1
					& \omega^3+\omega \\
    \omega^3+\omega^2
			& -\omega^3-1
				& 2\omega^2 
					&0 \\
    -1
			& \omega
				& 1 
					& -\omega^3+2\omega
    }.
  \]
  It has least denominator exponent $3$. Its $3$-, $4$-, and
  $5$-residues are:
  \[ \begin{split}
    &
    \rho_3(U) = \zmatrix{cccc}{
      1011 & 0111 & 0100 & 0010\\
      0110 & 1100 & 0101 & 1010\\
      1100 & 1001 & 0000 & 0000\\
      0001 & 0010 & 0001 & 1000
    },
    \\&
  \rho_4(U) = \zmatrix{cccc}{
      1010 & 0101 & 1010 & 0101\\
      1111 & 1111 & 0000 & 0000\\
      1111 & 1111 & 0000 & 0000\\
      1010 & 0101 & 1010 & 0101
    },\quad
  \rho_5(U) = 0.
  \end{split}
  \]
\end{example}

\section{Decomposition into two-level matrices}

Recall that a {\em two-level matrix} is an $n\times n$-matrix that
acts non-trivially on at most two vector components
{\cite{Nielsen-Chuang}}. If
\[U=\zmatrix{cc}{a&b\\c&d}
\]
is a $2\times 2$-matrix and $\jay\neq\ell$, we write $U\level{\jay,\ell}$ for the
two-level $n\times n$-matrix defined by
\[
  U\level{\jay,\ell}=\begin{blockarray}{cccccc}
    &\matindex{\cdots}& \matindex{\jay} &\matindex{\cdots}& \matindex{\ell} &\matindex{\cdots} \\
    \begin{block}{c(c|c|c|c|c)}
      \matindex{\vdots} & \bigI &&&& \\\cline{2-6}
      \matindex{\jay} & & a && b \\\cline{2-6}
      \matindex{\vdots} & && \bigI && \\\cline{2-6}
      \matindex{\ell} & & c && d &  \\\cline{2-6}
      \matindex{\vdots} &&&&& \bigI \\
    \end{block}
  \end{blockarray}\,,
\]
and we say that $U\level{\jay,\ell}$ is a two-level matrix {\em of type
  $U$}. Similarly, if $a$ is a scalar, we write
$a\level{\jay}$ for the one-level matrix
\[
 a\level{\jay} = \begin{blockarray}{cccc}
    &\matindex{\cdots}& \matindex{\jay} &\matindex{\cdots} \\
    \begin{block}{c(c|c|c)}
      \matindex{\vdots} & \bigI && \\\cline{2-4}
      \matindex{\jay} & & a & \\\cline{2-4}
      \matindex{\vdots} & && \bigI\\
    \end{block}
  \end{blockarray}\,,
\]
and we say that $a\level{\jay}$  is a one-level matrix {\em of type $a$}.

\begin{lemma}[Row operation]\label{lem-row}
  Let $u=(u_1,u_2)^T\in\Dw^2$ be a vector with denominator exponent
  $k>0$ and $k$-residue $\rho_k(u)=(x_1,x_2)$, such that $x_1\da x_1 =
  x_2\da x_2$.  Then there exists a sequence of matrices
  $U_1,\ldots,U_h$, each of which is $H$ or $T$, such that
  $v=U_1\cdots U_hu$ has denominator exponent $k-1$, or equivalently,
  $\rho_k(v)$ is defined and reducible.
\end{lemma}

\begin{proof}
  It can be seen from Table~\ref{tab-residue} that $x_1\da x_1$ is
  either $0000$, $1010$, or $0001$.
  \begin{itemize}
  \item Case 1: $x_1\da x_1 = x_2\da x_2 = 0000$. In this case,
    $\rho_k(u)$ is already reducible, and there is nothing to show.
  \item Case 2: $x_1\da x_1 = x_2\da x_2 = 1010$. In this case, we know
    from Table~\ref{tab-residue} that $x_1,x_2\in\s{0011,
      0110, 1100, 1001}$. In particular, $x_1$ is a cyclic permutation
    of $x_2$, say, $x_1=\omega^m x_2$. Let $v=HT^mu$. Then
    \[ \begin{split}
      \rho_k(\sqrt{2}\,v) &=
    \rho_k(\zmatrix{cc}{1&1\\1&-1}\zmatrix{cc}{1&0\\0&\omega^m}\zmatrix{c}{u_1\\u_2})
    \\&= \rho_k\zmatrix{c}{u_1+\omega^m u_2\\ u_1-\omega^m u_2} \\&=
    \zmatrix{c}{x_1+\omega^m x_2\\x_1-\omega^m x_2} =
    \zmatrix{c}{0000\\0000}.
    \end{split}
    \]
    This shows that $\rho_k(\sqrt{2}\,v)$ is twice reducible;
    therefore, $\rho_k(v)$ is defined and reducible as claimed.
  \item Case 3: $x_1\da x_1 = x_2\da x_2 = 0001$. In this case, we
    know from Table~\ref{tab-residue} that
    $x_1,x_2\in\s{0001,0010,0100,1000}\cup\s{0111,1110,1101,1011}$.
    If both $x_1,x_2$ are in the first set, or both are in the second
    set, then $x_1$ and $x_2$ are cyclic permutations of each other,
    and we proceed as in case 2. The only remaining cases are that
    $x_1$ is a cyclic permutation of $0001$ and $x_2$ is a cyclic
    permutation of $0111$, or vice versa. But then there exists some
    $m$ such that $x_1+\omega^mx_2=1111$.
    Letting $u'=HT^mu$, we have
    \[ \begin{split}
    \rho_k(\sqrt{2}\,u') &=
    \rho_k(\zmatrix{cc}{1&1\\1&-1}\zmatrix{cc}{1&0\\0&\omega^m}\zmatrix{c}{u_1\\u_2})
    \\&= \rho_k\zmatrix{c}{u_1+\omega^m u_2\\ u_1-\omega^m u_2} \\&=
    \zmatrix{c}{x_1+\omega^m x_2\\x_1-\omega^m x_2} =
    \zmatrix{c}{1111\\1111}.
    \end{split}
    \]
    Since this is reducible, $u'$ has denominator exponent $k$. Let
    $\rho_k(u')=(y_1,y_2)$. Because
    $\sqrt{2}\,y_1=\sqrt{2}\,y_2=1111$, we see from
    Table~\ref{tab-residue} that $y_1,y_2\in\s{0011,0110,1100,1001}$
    and $y_1\da y_1=y_2\da y_2=1010$. Therefore, $u'$ satisfies the
    condition of case 2 above. Proceeding as in case 2, we find $m'$
    such that $v=HT^{m'} u'=HT^{m'} HT^mu$ has denominator exponent
    $k-1$. This finishes the proof.\qedhere
  \end{itemize}
\end{proof}

\begin{lemma}[Column lemma]\label{lem-column}
  Consider a unit vector $u\in\Dw^n$, i.e., an $n$-dimensional column
  vector of norm 1 with entries from the ring $\Dw$. Then there exist
  a sequence $U_1,\ldots,U_h$ of one- and two-level unitary matrices
  of types $X$, $H$, $T$, and $\omega$ such that $U_1\cdots U_hu = e_1$,
  the first standard basis vector.
\end{lemma}

\begin{proof}
  The proof is by induction on $k$, the least denominator exponent of
  $u$. Let $u=(u_1,\ldots,u_n)^T$.
  \begin{itemize}
  \item Base case. Suppose $k=0$. Then $u\in\Zw^n$. Since by
    assumption $\norm{u}^2=1$, it follows by Lemma~\ref{lem-ring-norm}
    that $\weight{u}^2=1$.  Since $u_1,\ldots,u_n$ are elements of
    $\Zw$, their weights are non-negative integers. It follows that
    there is precisely one $\jay$ with $\weight{u_\jay}=1$, and
    $\weight{u_\ell}=0$ for all $\ell\neq \jay$. Let
    $u'=X\level{1,\jay}u$ if $\jay\neq 1$, and $u'=u$ otherwise. Now
    $u'_1$ is of the form $\omega^{-m}$, for some $m\in\s{0,\ldots,7}$,
    and $u'_\ell=0$ for all $\ell\neq 1$. We have
    $\omega^m\level{1}u'=e_1$, as desired.
  \item Induction step. Suppose $k>0$. Let $v=\sqrt{2}^ku\in\Zw^n$,
    and let $x=\rho_k(u) = \rho(v)$.  From $\norm{u}^2=1$, it follows
    that $\norm{v}^2 = v_1\da v_1+\ldots+v_n\da v_n = 2^k$. Taking
    residues of the last equation, we have
    \begin{equation}
      x_1\da x_1+\ldots+x_n\da x_n = 0000.
    \end{equation}
    It can be seen from Table~\ref{tab-residue} that each summand
    $x_\jay\da x_\jay$ is either $0000$, $0001$, or $1010$. Since their sum
    is $0000$, it follows that there is an even number of $\jay$ such
    that $x_\jay\da x_\jay=0001$, and an even number of $\jay$ such that
    $x_\jay\da x_\jay=1010$. 

    We do an inner induction on the number of irreducible components
    of $x$. If $x$ is reducible, then $u$ has denominator exponent
    $k-1$ by Corollary~\ref{cor-reducible3}, and we can apply the
    outer induction hypothesis. Now suppose there is some $\jay$ such
    that $x_\jay$ is irreducible; then $x_\jay\da x_\jay\neq 0000$ by
    Lemma~\ref{lem-reducible}. Because of the evenness property noted
    above, there must exist some $\ell\neq \jay$ such that $x_\jay\da
    x_\jay=x_\ell\da x_\ell$. Applying Lemma~\ref{lem-row} to
    $u'=(u_\jay,u_\ell)^T$, we find a sequence $\vec U$ of row
    operations of types $H$ and $T$, making $\rho_k(\vec Uu')$
    reducible. We can lift this to a two-level operation $\vec
    U\level{\jay,\ell}$ acting on $u$; thus $\rho_k(\vec
    U\level{\jay,\ell}u)$ has fewer irreducible components than
    $x=\rho_k(u)$, and the inner induction hypothesis applies.\qedhere
  \end{itemize}
\end{proof}

\begin{lemma}[Matrix decomposition]\label{lem-matrix-decomposition}
  Let $U$ be a unitary $n\times n$-matrix with entries in $\Dw$. Then
  there exists a sequence $U_1,\ldots,U_h$ of one- and two-level
  unitary matrices of types $X$, $H$, $T$, and $\omega$ such that
  $U=U_1\cdots U_h$.
\end{lemma}

\begin{proof}
  Equivalently, it suffices to show that there exist one- and
  two-level unitary matrices $V_1,\ldots,V_h$ of types $X$, $H$, $T$,
  and $\omega$ such that $V_h\cdots V_1U=I$.  This is an easy
  consequence of the column lemma, exactly as in
  e.g. {\cite[Sec.~4.5.1]{Nielsen-Chuang}}. Specifically, first use
  the column lemma to find suitable one- and two-level row operations
  $V_1,\ldots,V_{h_1}$ such that the leftmost column of $V_{h_1}\cdots
  V_1 U$ is $e_1$. Because $V_{h_1}\cdots V_1 U$ is unitary, it is of
  the form
  \[ \zmatrix{c|c}{1&0\\\hline \rule{0mm}{2.1ex}0&U'}.
  \]
  Now recursively find row operations to reduce $U'$ to the identity
  matrix. 
\end{proof}

\begin{example}\label{exa-decomp}
  We will decompose the matrix $U$ from Example~\ref{exa-k-residue}.
  We start with the first column $u$ of $U$:
	\[\begin{array}{c}
          \displaystyle
    u=\frac{1}{\sqrt{2}^3}\zmatrix{c}{
    -\omega^3+\omega -1\\
    \omega^2+\omega \\
    \omega^3+\omega^2\\
    -1
    }, \\\\[-1.5ex]
    \rho_3(u) = \zmatrix{c}{
    1011\\
    0110\\
    1100\\
    0001
    },\quad \rho_3(u_\jay\da u_\jay) = \zmatrix{c}{
    0001\\
    1010 \\
    1010\\
    0001
    }.
    \end{array}
    \]
	Rows 2 and 3 satisfy case 2 of Lemma~\ref{lem-row}. As they are not aligned, first apply 
	$T_{[2,3]}^3$ and then $H_{[2,3]}$. Rows 1 and 4 satisfy case 3. Applying $H_{[1,4]} T_{[1,4]}^2$, 
	the residues become
	$\rho_3(u'_1)=0011$ and $\rho_3(u'_4)=1001$, which requires applying $H_{[1,4]} T_{[1,4]}$. 
	We now have 
	\[\begin{array}{@{}c@{}}
          \displaystyle 
          H_{[1,4]} T_{[1,4]} H_{[1,4]} T_{[1,4]}^2 H_{[2,3]} T_{[2,3]}^3 u=
	  v=\frac{1}{\footnotesize\sqrt{2}^2}\small\zmatrix{c}{
    0 \\
		0 \\
    \omega^2{+}\omega  \\
    -\omega {+}1 
    }, \\\\[-1.5ex]
    \rho_2(v) = \zmatrix{c}{
    0000\\
    0000\\
    0110\\
    0011
    },\quad \rho_2(v_\jay\da v_\jay) = \zmatrix{c}{
    0000\\
    0000\\
    1010\\
    1010
    }.\end{array}\]
	Rows 3 and 4 satisfy case 2, while rows 1 and 2 are already reduced. We reduce rows 3 and 4 by
	applying $H_{[3,4]} T_{[3,4]}$.
	Continuing, the first column is completely reduced to $e_1$ by further applying 
	$\omega_{[1]}^7 X_{[1,4]} H_{[3,4]} T_{[3,4]}^3$.
	The  complete decomposition of $u$ is therefore given by 
	\[\begin{split}W_1={}&\omega_{[1]}^7  X_{[1,4]} H_{[3,4]} T_{[3,4]}^3
		H_{[3,4]} T_{[3,4]}\\&
		H_{[1,4]} T_{[1,4]} H_{[1,4]} T_{[1,4]}^2 H_{[2,3]} T_{[2,3]}^3.\end{split}\]
	Applying this to the original matrix $U$, we have $W_1U=$
	\[\small \frac{\small 1}{\small\sqrt{2}^3}\footnotesize\zmatrix{cccc}{
    \sqrt{2}^3             & 0      & 0                           & 0\\
    0 & \omega^3{-}\omega^2{+}\omega{+}1  & {-}\omega^2{-}\omega{-}1          & \omega^2\\
    0 & 0                           & \omega^3{+}\omega^2{-}\omega{+}1  & \omega^3{+}\omega^2{-}\omega{-}1\\
    0 & \omega^3{+}\omega^2{+}\omega{+}1 & \omega^2                    & \omega^3{-}\omega^2{+}1
    }.
\]
  Continuing with the rest of the columns, we 
  find $W_2 = \omega_{[2]}^6 H_{[2,4]} T_{[2,4]}^3 H_{[2,4]} T_{[2,4]}$,
  $W_3 = \omega_{[3]}^4 H_{[3,4]} T_{[3,4]}^3 H_{[3,4]}$, and
  $W_4=\omega_{[4]}^5$. We then have $U = W_1\da\, W_2\da\, W_3\da\,
  W_4\da$, or explicitly:
\[\begin{split} U ={}&  T_{[2,3]}^5 H_{[2,3]} T_{[1,4]}^6 H_{[1,4]} T_{[1,4]}^7 H_{[1,4]}\\&
					T_{[3,4]}^7 H_{[3,4]} 
					T_{[3,4]}^5 H_{[3,4]}
                                        X_{[1,4]} \omega_{[1]}\\&
  T_{[2,4]}^7 H_{[2,4]} T_{[2,4]}^5 H_{[2,4]} \omega_{[2]}^2
  H_{[3,4]} T_{[3,4]}^5 H_{[3,4]} \omega_{[3]}^4
  \omega_{[4]}^3.\end{split}\]
\end{example}

\section{Proof of Theorem~\ref{thm-main}}

\subsection{Equivalence of (a) and (b)}

First note that, since all the elementary Clifford+$T$ gates, as shown
in (\ref{eqn-generators}), take their matrix entries in
$\Dw=\Z[\frac1{\sqrt{2}},i]$, the implication (a) $\imp$ (b) is
trivial. For the converse, let $U$ be a unitary $2^n\times 2^n$ matrix
with entries from $\Dw$. By Lemma~\ref{lem-matrix-decomposition}, $U$
can be decomposed into one- and two-level matrices of types $X$, $H$,
$T$, and $\omega$. It is well-known that each such matrix can be
further decomposed into controlled-not gates and multiply-controlled
$X$, $H$, $T$, and $\omega$-gates, for example using Gray codes
{\cite[Sec.~4.5.2]{Nielsen-Chuang}}. But all of these gates have
well-known exact representations in Clifford+$T$ with ancillas, see
e.g. {\cite[Fig.~4(a) and Fig.~9]{AMMR12}} (and noting that a
controlled-$\omega$ gate is the same as a $T$-gate). This finishes the
proof of (b) $\imp$ (a).

\subsection{One ancilla is sufficient}

The final claim that needs to be proved is that a circuit for $U$ can
always be found using at most one ancilla. It is already known that
for $n>1$, an ancilla is sometimes necessary
{\cite{Kliuchnikov-et-al}}. To show that a single ancilla is
sufficient, in light of the above decomposition, it is enough to show
that the following can be implemented with one ancilla:
\begin{enumerate}\alphalabels
\item a multiply-controlled $X$-gate;
\item a multiply-controlled $H$-gate;
\item a multiply-controlled $T$-gate.
\end{enumerate}
We first recall from {\cite[Fig.~4(a)]{AMMR12}} that a
singly-controlled Hadamard gate can be decomposed into Clifford+$T$
gates with no ancillas:
\[ \m{\begin{qcircuit}[scale=0.5] 
    \grid{2}{0,1};
    \controlled{\gate{$H$}}{1,0}{1};
  \end{qcircuit}}
=\!\!\!
  \m{\begin{qcircuit}[scale=0.5]
    \grid{11}{0,1};
    \gate{$S$}{1,0};
    \gate{$H$}{2.5,0};
    \gate{$T$}{4,0};
    \controlled{\notgate}{5.5,0}{1};
    \widegate{$T\da$}{.55}{7,0};
    \gate{$H$}{8.5,0};
    \widegate{$S\da$}{.55}{10,0};
  \end{qcircuit}.}
\]
We also recall that an $n$-fold controlled
$iX$-gate can be represented using $O(n)$ Clifford+$T$ gates with no
ancillas. Namely, for $n=1$, we have
\[
\m{\begin{qcircuit}[scale=0.5]
    \grid{2}{1,2};
    \controlled{\widegate{$iX$}{.65}}{1,1}{2};
  \end{qcircuit}}
  =\!\!\!
  \m{\begin{qcircuit}[scale=0.5]
    \grid{3.5}{1,2};
    \gate{$S$}{1,2};
    \controlled{\notgate}{2.5,1}{2};
  \end{qcircuit},}
\]
and for $n\geq 2$, we can use
\[
    \m{\begin{qcircuit}[scale=0.5]
    \grid{2}{1,2,3,4,5};
    \draw(0.4,2.7) node[anchor=center]{$\vdots$};
    \draw(0.4,4.7) node[anchor=center]{$\vdots$};
    \controlled{\widegate{$iX$}{.65}}{1,1}{2,3,4,5};
    \draw(1.6,2.7) node[anchor=center]{$\vdots$};
    \draw(1.6,4.7) node[anchor=center]{$\vdots$};
  \end{qcircuit}}
  =\!\!\!
  \m{\begin{qcircuit}[scale=0.5]
    \gridx{0.1}{13.5}{1,2,3,4,5};
    \gate{$H$}{1.1,1};
    \widegate{$T\da$}{.55}{2.5,1};
    \controlled{\notgate}{3.75,1}{2,3};
    \gate{$T$}{5,1};
    \controlled{\notgate}{6.25,1}{4,5};
    \widegate{$T\da$}{.55}{7.5,1};
    \controlled{\notgate}{8.75,1}{2,3};
    \gate{$T$}{10,1};
    \controlled{\notgate}{11.25,1}{4,5};
    \gate{$H$}{12.5,1};
    \draw(3.25,2.7) node[anchor=center]{$\vdots$};
    \draw(4.25,2.7) node[anchor=center]{$\vdots$};
    \draw(5.75,4.7) node[anchor=center]{$\vdots$};
    \draw(6.75,4.7) node[anchor=center]{$\vdots$};
    \draw(8.25,2.7) node[anchor=center]{$\vdots$};
    \draw(9.25,2.7) node[anchor=center]{$\vdots$};
    \draw(10.75,4.7) node[anchor=center]{$\vdots$};
    \draw(11.75,4.7) node[anchor=center]{$\vdots$};
  \end{qcircuit},}
\]
with further decompositions of the multiply-controlled not-gates as in
{\cite[Lem.~7.2]{Barenco-etal-1995}} and
{\cite[Fig.~4.9]{Nielsen-Chuang}}.  We then obtain the following
representations for (a)--(c), using only one ancilla:
\[
(a)
  \m{\begin{qcircuit}[scale=0.47]
    \grid{2}{0.5,2,3};
    \draw(0.4,2.7) node[anchor=center]{$\vdots$};
    \controlled{\gate{$X$}}{1,0.5}{2,3};
    \draw(1.6,2.7) node[anchor=center]{$\vdots$};
  \end{qcircuit}}
  =\!\!
  \m{\begin{qcircuit}[scale=0.47]
    \grid{8.5}{0,2,3};
    \gridx{1}{7.5}{1};
    \draw(1,2.7) node[anchor=center]{$\vdots$};
    \init{$0$}{1,1};
    \controlled{\widegate{$iX$}{.75}}{2.5,1}{2,3};
    \controlled{\gate{$X$}}{4.25,0}{1};
    \controlled{\widegate{$-iX$}{.75}}{6,1}{2,3};
    \term{$0$}{7.5,1};
    \draw(7.5,2.7) node[anchor=center]{$\vdots$};
  \end{qcircuit}}
\]\[
(b)
  \m{\begin{qcircuit}[scale=0.47]
    \grid{2}{0.5,2,3};
    \draw(0.4,2.7) node[anchor=center]{$\vdots$};
    \controlled{\gate{$H$}}{1,0.5}{2,3};
    \draw(1.6,2.7) node[anchor=center]{$\vdots$};
  \end{qcircuit}}
  =\!\!
  \m{\begin{qcircuit}[scale=0.47]
    \grid{8.5}{0,2,3};
    \gridx{1}{7.5}{1};
    \draw(1,2.7) node[anchor=center]{$\vdots$};
    \init{$0$}{1,1};
    \controlled{\widegate{$iX$}{.75}}{2.5,1}{2,3};
    \controlled{\gate{$H$}}{4.25,0}{1};
    \controlled{\widegate{$-iX$}{.75}}{6,1}{2,3};
    \term{$0$}{7.5,1};
    \draw(7.5,2.7) node[anchor=center]{$\vdots$};
  \end{qcircuit}}
\]\[
(c)
  \m{\begin{qcircuit}[scale=0.47]
    \grid{2}{0.5,2,3};
    \draw(0.4,2.7) node[anchor=center]{$\vdots$};
    \controlled{\gate{$T$}}{1,0.5}{2,3};
    \draw(1.6,2.7) node[anchor=center]{$\vdots$};
  \end{qcircuit}}
  =\!\!\!
  \m{\begin{qcircuit}[scale=0.47]
    \grid{8.5}{1,2,3};
    \gridx{1}{7.5}{0};
    \draw(1,2.7) node[anchor=center]{$\vdots$};
    \init{$0$}{1,0};
    \controlled{\widegate{$iX$}{.75}}{2.5,0}{1,2,3};
    \gate{$T$}{4.25,0}{1};
    \controlled{\widegate{$-iX$}{.75}}{6,0}{1,2,3};
    \term{$0$.}{7.5,0};
    \draw(7.5,2.7) node[anchor=center]{$\vdots$};
  \end{qcircuit}}
\]

\begin{remark}
  The fact that one ancilla is always sufficient in
  Theorem~\ref{thm-main} is primarily of theoretical interest. In
  practice, one may assume that on most quantum computing
  architectures, ancillas are relatively cheap. Moreover, the use of
  additional ancillas can significantly reduce the size and depth
  of the generated circuits (see e.g.~\cite{Selinger-toffoli}).
\end{remark}

\section{The no-ancilla case}

\begin{lemma}\label{lem-det1}
  Under the hypotheses of Theorem~\ref{thm-main}, assume that $\det
  U=1$. Then $U$ can be exactly represented by a Clifford+$T$ circuit
  with no ancillas.
\end{lemma}

\begin{proof}
  This requires only minor modifications to the proof of
  Theorem~\ref{thm-main}. First observe that whenever an operator of
  the form $HT^m$ was used in the proof of Lemma~\ref{lem-row}, we can
  instead use $T^{-m}(iH)T^m$ without altering the rest of the
  argument. In the base case of Lemma~\ref{lem-column}, the operator
  $X\level{1,\jay}$ can be replaced by $iX\level{1,\jay}$. Also, in the
  base case of Lemma~\ref{lem-column}, whenever $n\geq 2$, the
  operator $\omega\level{1}$ can be replaced by $W\level{1,2}$, where
  \[ W = \zmatrix{cc}{\omega&0\\0&\omega^{-1}}.
  \]
  Therefore, the decomposition of Lemma~\ref{lem-matrix-decomposition}
  can be performed so as to yield only two-level matrices of types
  \begin{equation}\label{eqn-det1}
    iX,\quad T^{-m}(iH)T^m,\quad \mbox{and}\ W,
  \end{equation}
  plus at most one one-level matrix of type $\omega^m$. But since
  all two-level matrices of types (\ref{eqn-det1}), as well as $U$
  itself, have determinant 1, it follows that $\omega^m = 1$. We
  finish the proof by observing that the multiply-controlled operators
  of types (\ref{eqn-det1}) possess ancilla-free Clifford+$T$
  representations, with the latter two given by
  \[
  \m{\begin{qcircuit}[scale=0.47]
      \grid{5}{1,2,3};
      \draw(1,2.7) node[anchor=center]{$\vdots$};
      \controlled{\widegate{$T^{-m}(iH)T^m$}{2}}{2.5,1}{2,3};
      \draw(4,2.7) node[anchor=center]{$\vdots$};
    \end{qcircuit}}
  =\!\!
  \m{\begin{qcircuit}[scale=0.47]
      \grid{14.9}{1,2,3};
      \draw(1,2.7) node[anchor=center]{$\vdots$};
      \widegate{$T^m$}{.55}{1,1};
      \gate{$S$}{2.5,1};
      \gate{$H$}{4,1};
      \gate{$T$}{5.5,1};
      \controlled{\widegate{$iX$}{.75}}{7.25,1}{2,3};
      \widegate{$T\da$}{.55}{9,1};
      \gate{$H$}{10.5,1};
      \widegate{$S\da$}{.55}{12,1};
      \widegate{$T^{-m}$}{.75}{13.75,1};
      \draw(13.9,2.7) node[anchor=center]{$\vdots$};
    \end{qcircuit}}
  \]\[
  \m{\begin{qcircuit}[scale=0.47]
      \grid{2}{1,2,3};
      \draw(0.4,2.7) node[anchor=center]{$\vdots$};
      \controlled{\widegate{$W$}{.6}}{1,1}{2,3};
      \draw(1.6,2.7) node[anchor=center]{$\vdots$};
    \end{qcircuit}}
  =\!\!
  \m{\begin{qcircuit}[scale=0.47]
      \grid{7.5}{1,2,3};
      \draw(0.4,2.7) node[anchor=center]{$\vdots$};
      \controlled{\widegate{$iX$}{.75}}{1.25,1}{2,3};
      \gate{$T$}{3,1};
      \controlled{\widegate{$-iX$}{.75}}{4.75,1}{2,3};
      \widegate{$T\da$}{.55}{6.5,1};
      \draw(7.1,2.7) node[anchor=center]{$\vdots$};
    \end{qcircuit}}
  \]
\end{proof}

As a corollary, we obtain a characterization of the $n$-qubit
Clifford+$T$ group (with no ancillas) for all $n$:
  
\begin{corollary}\label{cor-noancilla}
  Let $U$ be a unitary $2^n\times 2^n$ matrix. Then the following are
  equivalent:
\begin{enumerate}\alphalabels
\item[(a)] $U$ can be exactly represented by a quantum
  circuit over the Clifford+$T$ gate set on $n$ qubits with no ancillas.
\item[(b)] The entries of $U$ belong to the ring
  $\Z[\frac1{\sqrt{2}},i]$, and:
  \begin{itemize}
    \item $\det U=1$, if $n\geq 4$;
    \item $\det U\in\s{-1,1}$, if $n=3$;
    \item $\det U\in\s{i,-1,-i,1}$, if $n=2$;
    \item $\det U\in\s{\omega,i,\omega^3,-1,\omega^5,-i,\omega^7,1}$,
      if $n\leq 1$.
    \end{itemize}
  \end{enumerate}
\end{corollary}

\begin{proof}
  For (a) $\imp$ (b), it suffices to note that each of the generators
  of the Clifford+$T$ group, regarded as an operation on $n$ qubits,
  satisfies the conditions in (b). For (b) $\imp$ (a), let us define
  for convenience $d_0=d_1=\omega$, $d_2=i$, $d_3=-1$, and $d_n=1$ for
  $n\geq 4$. First note that for all $n$, the Clifford+$T$ group on
  $n$ qubits (without ancillas) contains an element $D_n$ whose
  determinant is $d_n$, namely $D_n=I$ for $n\geq 4$, $D_3=T\otimes
  I\otimes I$, $D_2=T\otimes I$, $D_1=T$, and $D_0=\omega$. Now
  consider some $U$ satisfying (b). By assumption, $\det U=d_n^m$
  for some $m$. Let $U'=UD_n^{-m}$, then $\det U'=1$. By
  Lemma~\ref{lem-det1}, $U'$, and therefore $U$, is in the
  Clifford+$T$ group with no ancillas.
\end{proof}

\begin{remark}
  Note that the last condition in Corollary~\ref{cor-noancilla},
  namely that $\det U$ is a power of $\omega$ for $n\leq 1$, is of
  course redundant, as this already follows from $\det
  U\in\Z[\frac1{\sqrt{2}},i]$ and $|\det U|=1$. We stated the
  condition for consistency with the case $n\geq 2$.
\end{remark}

\begin{remark}
  The situation of Theorem~\ref{thm-main} and
  Corollary~\ref{cor-noancilla} is analogous to the case of classical
  reversible circuits. It is well-known that the not-gate,
  controlled-not gate, and Toffoli gate generate all classical
  reversible functions on $n\leq 3$ bits. For $n\geq 4$ bits, they
  generate exactly those reversible boolean functions that define an
  {\em even permutation} of their inputs (or equivalently, those that
  have determinant 1 when viewed in matrix form) {\cite{Musset97}};
  the addition of a single ancilla suffices to recover all boolean
  functions.
\end{remark}

\section{Complexity}

The proof of Theorem~\ref{thm-main} immediately yields an algorithm,
albeit not a very efficient one, for synthesizing a Clifford+$T$
circuit with ancillas from a given operator $U$. We estimate the size
of the generated circuits.

We first estimate the number of (one- and two-level) operations
generated by the matrix decomposition of
Lemma~\ref{lem-matrix-decomposition}. The row operation from
Lemma~\ref{lem-row} requires only a constant number of
operations. Reducing a single $n$-dimensional column from denominator
exponent $k$ to $k-1$, as in the induction step of
Lemma~\ref{lem-column}, requires $O(n)$ operations; therefore, the
number of operations required to reduce the column completely is
$O(nk)$.

Now consider applying Lemma~\ref{lem-matrix-decomposition} to an
$n\times n$-matrix with least denominator exponent $k$. Reducing the
first column requires $O(nk)$ operations, but unfortunately, it may
{\em increase} the least denominator exponent of the rest of the
matrix, in the worst case, to $3k$. Namely, each row operation of
Lemma~\ref{lem-row} potentially increases the denominator exponent by
$2$, and any given row may be subject to up to $k$ row operations,
resulting in a worst-case increase of its denominator exponent from
$k$ to $3k$ during the reduction of the first column.  It follows
that reducing the second
column requires up to $O(3(n-1)k)$ operations, reducing the third column
requires up to $O(9(n-2)k)$ operations, and so on. Using the identity
$\sum_{j=0}^{n-1}3^j(n-j) = (3^{n+1}-2n-3)/4$, this results in a total of
$O(3^nk)$ one- and two-level operations for
Lemma~\ref{lem-matrix-decomposition}.

In the context of Theorem~\ref{thm-main}, we are dealing with $n$
qubits, i.e., a $2^n\times 2^n$-operator, which therefore decomposes
into $O(3^{2^n}k)$ two-level operations. Using one ancilla,
each two-level operation can be decomposed into $O(n)$ Clifford+$T$
gates, resulting in a total gate count of $O(3^{2^n}\!nk)$
elementary Clifford+$T$ gates.

\section{Future work}

As mentioned in the introduction, the algorithm arising out of the
proof of Theorem~\ref{thm-main} produces circuits that are very far
from optimal. This can be seen heuristically by taking any simple
Clifford+$T$ circuit, calculating the corresponding operator, and then
running the algorithm to re-synthesize a circuit. 

Moreover, it is unlikely that the algorithm is optimal even in the
asymptotic sense. The algorithm's worst case gate count of
$O(3^{2^n}\!nk)$ is separated from information-theoretic lower bounds
by an exponential gap. Specifically, the number of different unitary
$n$-qubit operators with denominator exponent $k$ can be bounded: for
$n\geq 1$, it is between $2^{2^{n-1}k}$ and
$2^{4^n(4+2k)}$. Therefore, such an operator carries between
$\Omega(2^nk)$ and $O(4^nk)$ bits of information.  Regardless of where
the true number falls within this spectrum, the resulting
information-theoretic lower bound for the number of elementary gates
required to represent such an operator is exponential, not
super-exponential, in $n$.

While the information-theoretic analysis does of course not imply the
existence of an asymptotically better synthesis algorithm, it
nevertheless suggests that it may be worthwhile to look for one. 

Given that the gate count estimate is dominated by the term $3^{2^n}\!$,
the most obvious target for improvement is the part of the algorithm
that causes this super-exponential blowup. As noted above, this blowup
is caused by the fact that row reductions that reduce the denominator
exponent of one column might simultaneously increase the denominator
exponent of the remaining columns.

\section{Acknowledgements}

This research was supported by NSERC.

This research was supported by the Intelligence Advanced Research
Projects Activity (IARPA) via Department of Interior National Business
Center contract number D11PC20168. The U.S.\ Government is authorized
to reproduce and distribute reprints for Governmental purposes
notwithstanding any copyright annotation thereon. Disclaimer: The
views and conclusions contained herein are those of the authors and
should not be interpreted as necessarily representing the official
policies or endorsements, either expressed or implied, of IARPA,
DoI/NBC, or the U.S.\ Government.

\bibliographystyle{abbrvunsrt}
\bibliography{nqubit}

\begin{thebibliography}{1}

\bibitem{Nielsen-Chuang}
M.~A. Nielsen and I.~L. Chuang.
\newblock {\em Quantum Computation and Quantum Information}.
\newblock Cambridge University Press, 2002.

\bibitem{Kliuchnikov-et-al}
V.~Kliuchnikov, D.~Maslov, and M.~Mosca.
\newblock Fast and efficient exact synthesis of single qubit unitaries
  generated by {Clifford} and {$T$} gates.
\newblock arXiv:1206.5236v2, June 2012.

\bibitem{AMMR12}
M.~Amy, D.~Maslov, M.~Mosca, and M.~Roetteler.
\newblock A meet-in-the-middle algorithm for fast synthesis of depth-optimal
  quantum circuits.
\newblock Version 2, arXiv:1206.0758v2, Aug. 2012.

\bibitem{Barenco-etal-1995}
A.~Barenco, C.~H. Bennett, R.~Cleve, D.~P. DiVincenzo, N.~Margolus, P.~Shor,
  T.~Sleator, J.~A. Smolin, and H.~Weinfurter.
\newblock Elementary gates for quantum computation.
\newblock {\em Physical Review A}, 52:3457--3467, 1995.
\newblock Available from arXiv:quant-ph/9503016v1.

\bibitem{Selinger-toffoli}
P.~Selinger.
\newblock Quantum circuits of {$T$}-depth one.
\newblock {\em Physical Review A}, 2013.
\newblock To appear. Available from arXiv:1210.0974.

\bibitem{Musset97}
J.~Musset.
\newblock G\'en\'erateurs et relations pour les circuits bool\'eens
  r\'eversibles.
\newblock Technical Report 97-32, Institut de Math\'ematiques de Luminy, 1997.
\newblock Available from http://iml.univ-mrs.fr/editions/.

\end{thebibliography}

\end{document}